\documentclass[aps, pra, twocolumn, superscriptaddress]{revtex4-2}
\usepackage[utf8]{inputenc}
\usepackage{amsmath, amsthm, amssymb}
\usepackage{amsthm}
\usepackage{amsfonts}
\usepackage{tikz}
\usepackage{braket}
\usepackage[english]{babel}
\usepackage{comment}
\usepackage{caption}
\usepackage{subfigure}
\usepackage{graphicx} 
\usepackage[colorlinks=true,linkcolor=blue,urlcolor=blue,citecolor=blue]{hyperref}
\newtheorem{prop}{Proposition}

\newtheorem{result}{Result}

\begin{document}
\title{Breaking absolute separability with quantum switch}

\author{Sravani Yanamandra}
\email{sravani.yanamandra@research.iiit.ac.in}

\author{P V Srinidhi}
\email{srinidhi.pv@research.iiit.ac.in}

\author{Samyadeb Bhattacharya}
\email{samyadeb.b@iiit.ac.in }

\author{Indranil Chakrabarty}
\email{indranil.chakrabarty@iiit.ac.in}
\affiliation{Centre for Quantum Science and Technology, International Institute of Information Technology Hyderabad, Gachibowli, Hyderabad-500032, Telangana, India.\\Center for Security, Theory and Algorithmic Research, International Institute of Information Technology Hyderabad, Gachibowli, Hyderabad-500032, Telangana, India.}

\author{Suchetana Goswami}
\email{suchetana.goswami@gmail.com}
\affiliation{Harish-Chandra Research Institute, A CI of Homi Bhabha National Institute, Chhatnag Road, Jhunsi, Allahabad 211 019, India.}

\begin{abstract}
    Absolute separable (AS) quantum states are those states from which it is impossible to create entanglement, even under global unitary operations. It is known from the resource theory of non-absolute separability that the
    set of absolute separable states forms a convex and compact set, and global unitaries are free operations. We show that the action of a quantum switch controlled by an ancilla qubit over the global unitaries can break this robustness of AS states and produce ordinary separable states. First, we consider bipartite qubit systems and find the effect of quantum switch starting from the states sitting on the boundary of the set of absolute separable states. As particular examples, we illustrate what happens to modified Werner states and Bell diagonal (BD) states. For the Bell diagonal states, we provide the structure for the set of AS BD  states and show how the structure changes under the influence of a switch. Further, we consider numerical generalization of the global unitary operations and show that it is always possible to take AS states out of the convex set under switching operations. We also generalized our results in higher dimensions.

\end{abstract}

\maketitle

\section{Introduction}
In 1935, Einstein, Podolsky, and Rosen (EPR) first pointed out the concept of non-local effect in quantum systems in terms of a paradox \cite{EPR_35}. Later in the same year, while trying to explain the paradox, Schr\"{o}dinger \cite{S_35} introduced the term ``entanglement'' and the concept of quantum steering \cite{S_35, WJD_07, JWD_07} which eventually leads to another quantum correlation i.e. Bell-nonlocality \cite{B_64}. From then on, in the literature of quantum information theory, the concept of entanglement as a quantum correlation has consistently been the most dominant one. It is not only of philosophical interest where it deals with deep concepts like realism and locality but also acts as a cardinal resource in various information processing tasks such as teleportation \cite{BBCJPW_93}, super dense coding \cite{BW_92} and many more \cite{BB_84, E_91}. In the recent developments in quantum computation, entanglement is an intricate resource for speed-ups compared to the corresponding classical counterparts \cite{LP_01, JL_03,V_13,GH_23} and also to minimize the effect of environmental noise \cite{DGPM_17}. When multipartite systems are considered, entanglement is one of the most basic forms of non-classical correlation that is also a resource. On the other hand, quantum steering is a stronger form of correlation than entanglement and the steerable states form a strict subset of entangled states. Finally, the Bell non-local states are a strict subset of steerable states. This hierarchy \cite{WJD_07} and application of all these non-local correlations are well explored in literature \cite{TR_11,BCWSW_12}. \\

Entangled states are those which are not separable in nature i.e. they can not be written as a convex mixture of product states of the corresponding subsystems. From the perspective of entanglement resource theory \cite{HHH_09}, the local operations and classical communication (LOCC) are the free operations, using which neither entanglement can be created from a separable state, nor can be increased. Hence in this resource theory, the separable states are the free states and they form a convex and compact set. Now it is natural to ask if one has free access to global operations, then is it possible to create entanglement from a separable state? The answer is in general affirmative but is not always true. There is a set of states which is so strongly separable that it is not possible to generate entangled states from them via any global unitary or even via the convex combination of them \cite{KZ_01, VAD_01, KC_01, LHL_03}. These states are called absolute separable (AS) states. Recently, in \cite{PMD_22} the resource theory of non-absolute separability is introduced where global unitaries are the free operations and the non-AS (NAS) states are the resourceful states. However, these absolute states are not restricted to only the entanglement-separability paradigm. There are other classes of absolute states such as absolute local states, absolute fully entangled fraction states, absolute conditional entropy non-negative states and many more \cite{PCG_17, BMRPMCJG_18, PMSCG_22}.\\

In the recent past, the literature of quantum information theory evidently illustrates that indefinite ordering of multiple causal events can give rise to advantages in terms of resource requirements in different quantum protocols. The concept of indefinite causal ordering was first pointed out in \cite{H_05, H_07} and was extended to resource theoretic frameworks in \cite{CDPV_13}. Exploiting this concept, the structure of quantum switch was introduced \cite{CDPV_13}, where we essentially consider an ancillary system that acts as a controller to the orders of the events. A stronger approach to this indefiniteness is by process matrix formalism introduced in \cite{OCB_12}. Indefinite causal ordering has played a significant role in increasing the efficacy of various information processing activities. These include winning non-local games \cite{OCB_12}, testing properties of quantum channel \cite{C_12}, minimizing quantum communication complexity \cite{GFAC_16}, improving quantum communication \cite{ESC_18, CBBGARSAK_21,Mitra_2023}, increasing the performance of quantum algorithm \cite{ACB_14}, activating non-Markovianity \cite{MB_22,agm_23} and many more. Recent developments in this domain are centered around providing experimental evidence \cite{PMACADHRBW_15, RRFAZPBW_17, GGKCBRW_18}. \\

In this work, we try to connect the effect of indefinite causal order with the possibility of generating some resourceful state starting from AS states using global unitary operations. From the resource theoretic point of view, if we have global unitary operations for free, then can we take states out of the convex set of AS states? We answer the question affirmatively with the help of a quantum switch. Here, we consider that the two global unitary operations represented by two channels $(\mathcal{N}_1, \mathcal{N}_2)$ are acting sequentially. The power of switching is induced by introducing the ancillary system, which dictates that with some probability $\mathcal{N}_2$ acts after $\mathcal{N}_1$ and $\mathcal{N}_1$ acts after $\mathcal{N}_2$ with rest of the probability. Next, to ensure the effect of indefinite causal ordering, we consider the superposition of these two scenarios. In the beginning, we consider AS states lying on the boundary of the convex set and show that by using global unitary with switching action, one can take these states out of this set. Next, we consider two-qubit modified Werner state and show how the range of the AS states changes with the value of the parameters in the presence of quantum switch. We notice that it is possible to make even a maximally mixed state resourceful by using the switching operation of suitably chosen unitary. In \cite{LC_10}, the authors identify the geometry of the separable BD states. Here, we extend the results for AS states and show how the structure changes in the presence of switching on global unitary operations. Additionally, we generalize our results numerically by Haar uniformly generating random global unitary matrices. We show our protocol is effective even in higher dimensions.

\section{Entanglement, separability and absolute separability}
Quantum entanglement is the most powerful resource not only in information processing tasks but also in quantum computation. Here we begin our study in the backdrop of two-qubit systems. Let the bipartite state shared between two parties $A$ and $B$ be $\rho_{AB}$. The composite system is associated with the Hilbert space $\mathcal{H}_{AB} = \mathcal{H}_A \otimes \mathcal{H}_B$. We denote the space of bounded linear operators acting on $\mathcal{H}$ by $\mathcal{B}(\mathcal{H})$ and the set of density operators by $\mathcal{D}(\mathcal{H})$. Now if $\rho_{AB}= \sum_i p_i \rho_A^i \otimes \rho_B^i$, then the state is a separable state. Here, $\rho_A^i$ and $\rho_B^i$ are the reduced states of the subsystems $A$ and $B$ respectively. When a quantum state cannot be represented by the above-mentioned separable form, it is entangled. This non-separability of quantum states acts as a resource in quantum information theory. In $2\otimes2$ and $2\otimes3$ dimensions, there exists a necessary and sufficient criterion to detect entanglement \cite{P_96, HHH_96}. It is based on the partial transposition of the corresponding density matrices. According to this criterion, all entangled states have negative partial transposition (NPT) while the separable states have positive partial transposition (PPT) in the above-mentioned dimension. As mentioned earlier, it is not possible to create entanglement from a separable state via LOCC. On the other hand, if one has access to global operations, then it is possible to get an entangled state from a separable one and hence it can be used as a resource in different protocols. \\
A quantum state is called absolute separable if any global unitary operation on the state does not result in entanglement. In other words, a bipartite absolute separable states  $\rho_{AB} \in \mathcal{H}_A \otimes \mathcal{H}_B$ is a state such that for all unitary matrices (or even with their successsive applications) $\mathcal{U} : \mathcal{H}_A \otimes \mathcal{H}_B \longrightarrow \mathcal{H}_A \otimes \mathcal{H}_B$, we have $\mathcal{U} \rho_{AB} \mathcal{U}^\dagger \in \mathcal{S}_{asep}$ where, $\mathcal{S}_{asep}$ is the set of all absolute separable states. The set of absolute separable states forms a convex and compact set \cite{GCM_14} and they play the role of free states in the resource theory of NAS \cite{PMD_22}. In general, it is a difficult task to detect AS states but for qubit-qudit dimension one can have a criterion based on the eigenvalues of the density matrix \cite{Hi_07, S_09, J_13}. In this case, the $2d$ eigenvalues, arranged in a non-increasing order, are $\{\lambda_i^{\downarrow}\}$ with $\sum_i \lambda_i^{\downarrow} =1$. Then a state is AS if and only if the eigenvalues accept the following condition,
\begin{equation}
    \lambda_1^{\downarrow} - \lambda_{2d-1}^{\downarrow} - 2\sqrt{\lambda_{2d-2}^{\downarrow} \lambda_{2d}^{\downarrow}} \leq 0
    \label{AS_eigenvalue}
\end{equation}
The equality of the above equation holds for the states on the boundary of the convex set containing AS states \cite{HMD_21}. From the above inequality, the following consequences emerge immediately. 
\begin{prop}
    In $2\otimes d$ dimension, there exists no AS state with rank $(2d-2)$. \cite{HMD_21}
    \label{prop1}
\end{prop}
\begin{proof}
    The proof of the statement follows from \cite{HMD_21}. Note that, if we have a bipartite state $\rho_{AB}$ with rank $r(\rho_{AB}) \leq (2d-2)$, then we can write the $2d$ number of eigenvalues of the density matrix of $\rho_{AB}$ as, $\{\lambda_1^{\downarrow}, \lambda_2^{\downarrow}, \cdots, \lambda_{2d-2}^{\downarrow}, 0, 0\}$ without any loss of generality (the eigenvalues are arranged in a non-increasing manner). Hence using Eq.(\ref{AS_eigenvalue}), we have the condition on the largest eigenvalue as, $\lambda_1^{\downarrow} \leq 0$, which is not possible. So, one can conclude that there exists no rank $(2d-2)$ AS state in $2\otimes d$ dimension.
\end{proof}
\noindent Now let us consider the case when the rank of the state is $r(\rho_{AB})=(2d-1)$ and the eigenvalues can be written as, $\{\lambda_1^{\downarrow}, \lambda_2^{\downarrow}, \cdots, \lambda_{2d-1}^{\downarrow}, 0\}$. In this case, we have $\lambda_1^{\downarrow} \leq \lambda_{2d-1}^{\downarrow}$ from Eq.(\ref{AS_eigenvalue}). As the eigenvalues are already arranged in non-increasing order, we can set, $\lambda_1^{\downarrow}=\lambda_2^{\downarrow} \cdots= \lambda_{2d-1}^{\downarrow} = \lambda$ (say). Hence, $\lambda=\frac{1}{2d-1}$ as, $\sum_i \lambda_i^{\downarrow} =1$. For example, in a bipartite qubit system, we can have a state $\rho_{AB}=\frac{1}{3} (\ket{00}\bra{00}+\ket{01}\bra{01}+\ket{10}\bra{10})$. Note that from the eigenvalue condition, it is clear that this state resides on the boundary of the convex set in that dimension. At the beginning of our calculation, we consider such types of states and see the action of switching unitary on them. 

\section{Application of quantum switch on global unitary operation}

In this section, we briefly discuss the action of quantum switch as defined in the literature of quantum information theory and further illustrate the action of the same while the channels under consideration are global unitaries. Quantum switch is a higher-order map that takes two or more channels as inputs and then outputs the superposition of their orders based on the state of the control qubit. With a quantum switch, one can in principle control the action of the channel with an ancillary system, which is the control. As mentioned in the introduction, here also we consider two channels $\mathcal{N}_1$ and $\mathcal{N}_2$. When the state of the control qubit is $\ket{0}\bra{0}$, the channel $\mathcal{N}_2$ acts before the channel $\mathcal{N}_1$, the action on the system can be written mathematically as, $(\mathcal{N}_1 \circ \mathcal{N}_2) \rho_{AB}$. On the other hand, when the ancillary qubit is in state $\ket{1}\bra{1}$, the action of the channels on the system is given by, $(\mathcal{N}_2 \circ \mathcal{N}_1) \rho_{AB}$. Now we consider the superposition of these two situations to obtain the final switching action. 

Now let us consider that the Kraus operators corresponding to the channels  $\mathcal{N}_1$ and $\mathcal{N}_2$ are given by, $\{K_{i}^{(1)}\}$ and $\{K_{j}^{(2)}\}$ respectively, with $\sum_i {K^{(1)}_i}^{\dagger} K^{(1)}_i = \openone$ and $\sum_j {K^{(2)}_j}^{\dagger} K^{(2)}_j = \openone$. Hence the joint Kraus operator representing the complete action of the switch can be written as, 
 \begin{equation}
    W_{ij} = K_{i}^{(1)} \cdot K_{j}^{(2)} \otimes \ket{0}\bra{0} + K_{j}^{(2)} \cdot K_{i}^{(1)} \otimes \ket{1}\bra{1}.
    \label{switch_Kraus}
\end{equation}
Hence the final joint system-ancilla state can be written as, 
\begin{equation}
    S(\mathcal{N}_1,\mathcal{N}_2)(\rho \otimes \rho_c) = \sum_{i,j} W_{ij}(\rho_{AB} \otimes \rho_c)W_{ij}^{\dagger}.
    \label{switchaction}
\end{equation}
Here, $\rho_c$ is the initial state of the control qubit. In the end, the control qubit is measured in the coherent basis i.e. $\{\ket{+},\ket{-}\}$ and the final state of the system is obtained corresponding to each outcome as, 
\begin{equation}
    \rho_f = \frac{1}{n}~ _c\langle\pm|S(\mathcal{N}_1,\mathcal{N}_2)(\rho_{AB} \otimes \rho_c)|\pm\rangle_c
    \label{FS}
\end{equation}
with $\ket{\pm}=\frac{1}{2}(\ket{0}\pm\ket{1})$ and $n$ is the constant for normalization. 
In our case, the two channels under consideration are two global unitary matrices acting on the bipartite system. Let the unitary matrices be $\mathcal{U}_1$ and $\mathcal{U}_2$ and hence the action of the switch can be simplified as given in the following Proposition.
\begin{prop}
The final state of the system after the measurement, considering both the control qubit and the projective measurement on the control qubit in the $\ket{+}$ basis becomes,
\begin{equation}
    \rho_f=\frac{\mathcal{L}_s \rho_{AB} \mathcal{L}_s^{\dagger}}{ Tr[\mathcal{L}_s \rho_{AB} \mathcal{L}_s^{\dagger}]}
    \label{finalstate_switch}
\end{equation}
where, the switching action of the unitary matrices can be seen as, $\mathcal{L}_s=\frac{1}{2}(\mathcal{U}_1 \mathcal{U}_2+\mathcal{U}_2 \mathcal{U}_1)$.
\label{prop2}
\end{prop}
\begin{proof}
    We start the proof by incorporating Eq. (\ref{switchaction}) into Eq. (\ref{FS}) and therefore we have,
    \begin{equation}
        \rho^{\prime}_{f} =~ _c\bra{+}\sum_{i,j} W_{ij}(\rho_{AB} \otimes \rho_c)W_{ij}^{\dagger}|\ket{+}_c.
    \end{equation}
    Evidently, $\rho_f= \rho^{\prime}_{f}/n$ with $n= Tr [\rho^{\prime}_{f}]$. Let us now simplify the preliminary calculation by considering that the initial state is given by $\rho_c=\ket{+}\bra{+}$. The calculation elaborated in the following is for the case when the control qubit is measured in the $\ket{+}$ state. Hence we have,
    \begin{eqnarray}
        \rho^{\prime}_{f} &=~& _c\bra{+} (\mathcal{U}_1 \mathcal{U}_2 \otimes \ket{0}\bra{0} + \mathcal{U}_2 \mathcal{U}_1 \otimes \ket{1}\bra{1}) \nonumber\\
        &&(\rho_{AB}\otimes \ket{+}\bra{+}) \nonumber\\
        &&(\mathcal{U}^{\dagger}_2 \mathcal{U}^{\dagger}_1 \otimes \ket{0}\bra{0} + \mathcal{U}^{\dagger}_1 \mathcal{U}^{\dagger}_2 \otimes \ket{1}\bra{1}) \ket{+}_c \nonumber\\
        &=& (\mathcal{U}_1 \mathcal{U}_2~|\braket{0|+}|^2 + \mathcal{U}_2 \mathcal{U}_1~|\braket{1|+}|^2 ) \rho_{AB} \nonumber\\
        && (\mathcal{U}^{\dagger}_2 \mathcal{U}^{\dagger}_1~ |\braket{0|+}|^2 + \mathcal{U}^{\dagger}_1 \mathcal{U}^{\dagger}_2~|\braket{1|+}|^2 ) \nonumber\\
        &=& \mathcal{L}_s \rho_{AB} \mathcal{L}^{\dagger}_s, 
    \end{eqnarray}
    with $\mathcal{L}_s=\frac{1}{2}(\mathcal{U}_1 \mathcal{U}_2+\mathcal{U}_2 \mathcal{U}_1)$. Now it is clear that $n=Tr[\mathcal{L}_s \rho_{AB} \mathcal{L}^{\dagger}_s]$ and the final state of the system after the switching action is given by the expression in Eq. (\ref{finalstate_switch}). Hence the claim.
\end{proof}


\section{Breaking absolute separability with quantum switch}

In this section, we discuss the results of ``breaking" absolute separability via quantum switch. Here, by ``breaking" we mean having a separable state from a AS state via global unitary. We start by considering a very simple case to see if it is at all possible. In a bipartite qubit system, let us consider an AS state. Note that, the rank of the state can either be $3$ or $4$.

\begin{result}
    Let us consider a rank $3$ AS state $\rho_{AB}=\frac{1}{3} (\ket{00}\bra{00}+\ket{01}\bra{01}+\ket{10}\bra{10})$ which lies on the boundary of the convex set in $2\otimes 2$ dimension. Under the switching action of two global unitaries, one being CNOT gate and the other given by, 
    \begin{equation}
        U = 
    \begin{pmatrix}
        cos\theta & 0 & 0 & sin\theta\\
        0 & cos\theta & sin\theta & 0\\
        0 & sin\theta & -cos\theta & 0\\
        sin\theta & 0 & 0 & -cos\theta
    \end{pmatrix},
    \label{unitary_theta}
    \end{equation}
    the resulting state is not an AS state.
\end{result}

\begin{proof}
    We know that the CNOT gate is described by,
    \begin{equation}
    V_{CNOT}= \left(
    \begin{array}{cccc}
     1 & 0 & 0 & 0 \\
     0 & 1 & 0 & 0 \\
     0 & 0 & 0 & 1 \\
     0 & 0 & 1 & 0 \\
    \end{array}
    \right)
    \end{equation}
    Now, following Eq.(\ref{finalstate_switch}) one can obtain the final state where $\mathcal{U}_1$ and $\mathcal{U}_2$ are CNOT and $U$. The eigenvalues of the state arranged in a non-increasing manner are given by, $\{ \frac{4}{3(3+cos2\theta)}, \frac{1}{3}, \frac{2(1+cos2\theta)}{3(3+cos2\theta)},0 \}$. Note that, with these eigenvalues, the inequality in Eq. (\ref{AS_eigenvalue}) reduces to $\cos{(2\theta)}\geq 1$, which is itself contradictory. Hence the final state is residing outside the set of AS states in the given dimension. Now with PPT criterion, we verify that the state is a separable state. Hence we obtain a separable state starting from an AS state.
\end{proof}

\subsection*{Results on modified Werner states}

As can be seen from the above result, the AS state residing on the boundary of the set can be taken out from the set by using quantum switch. Now let us consider a more general form of a state in $2\otimes2$ dimension, the modified Werner state \cite{W_89} and it can be written as, 
\begin{equation}
    \rho_W = p \ket{\xi}\bra{\xi}+ \frac{1-p}{4} \mathcal{I}_{4 \times 4}
    \label{mod_Wstate}
\end{equation}
where $\ket{\xi}$ being a pure state of the form, $\ket{\xi} = \cos{\gamma}\ket{00} + e^{i\phi} \sin{\gamma}\ket{11}$, with $0 \leq p \leq 1$, $0 \leq \gamma \leq \pi$ and $0 \leq \phi \leq 2\pi$. It is already well established that the state in Eq. (\ref{mod_Wstate}) is AS for the range $0\leq p \leq \frac{1}{3}$ and it is separable in the region $\frac{1}{3}\leq p\leq \frac{1}{1+2\sin{2 \gamma}}$. Next we move a step further and show that it is possible to achieve separability for a larger region of $p$ when subjected to a global unitary channel via switching order. 
    
\begin{result}
    If the initial state is $\rho_W$ given in Eq.(\ref{mod_Wstate}) and it undergoes the switching operation between two unitary channels, one being CNOT and the other given in Eq. (\ref{unitary_theta}), then it is possible to achieve separability from the region of AS states.
\end{result}

\begin{figure}[htp]
\centering
\fbox{
\subfigure[The eigenvalues of the final state after the switching action on the state in Eq. (\ref{mod_Wstate}) are plotted against the global unitary parameter, $\theta$. Here we consider $p=0.15$ for which  the initial state is AS.]{\includegraphics[scale=0.25]{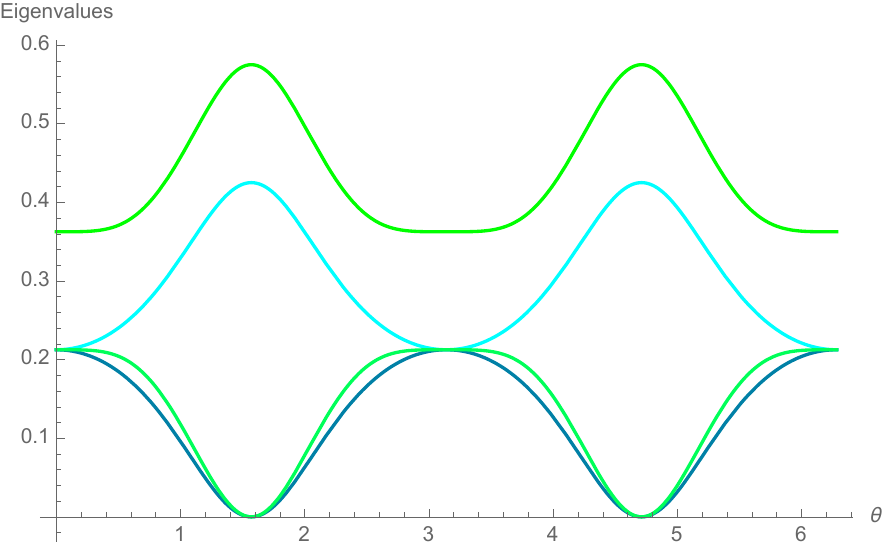}}
\qquad
\subfigure[The same eigenvalues are plotted against the noise mixing parameter $p$ ($p$ is varying in the range of AS states) for a fixed $\theta = 1.2$. ]{\includegraphics[scale=0.22]{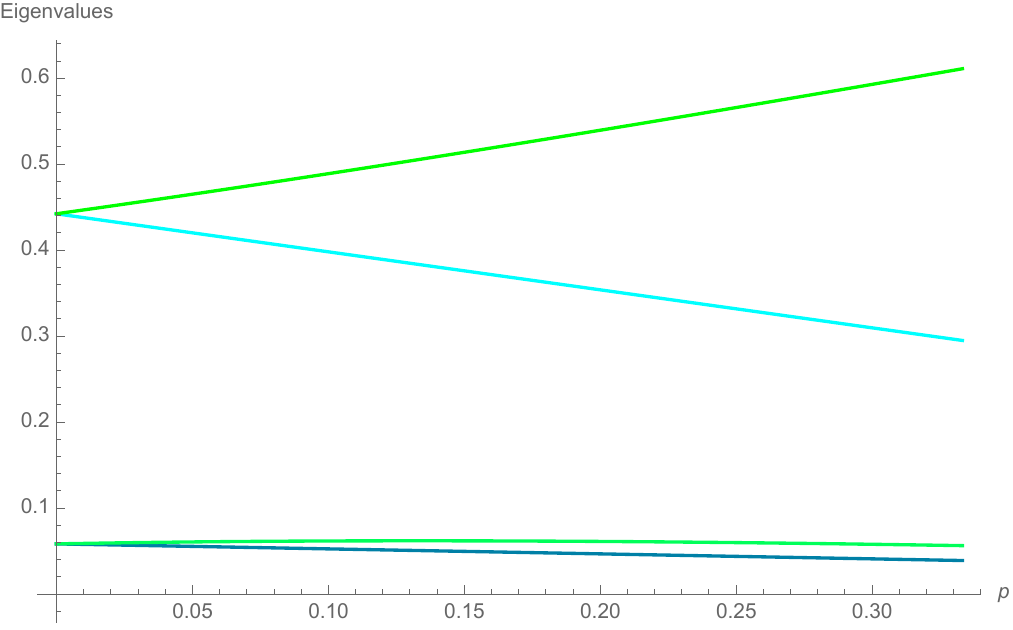}}}
\caption{\footnotesize{The eigenvalues of the final state}}
\label{eval_theta_modW}
\end{figure}

   \noindent  To prove the statement, we only need to demonstrate an example in which the stated transformation is possible. To begin the demonstration, we observe that the eigenvalues of the initial state in Eq. (\ref{mod_Wstate}) are given by $\{ \frac{1+3p}{4}, \frac{1-p}{4},\frac{1-p}{4},\frac{1-p}{4} \}$, arranged in the non-increasing order. Note that, the eigenvalues are independent of the entanglement content of the pure state $\ket{\xi}$. Now using the process prescribed via Eq. (\ref{finalstate_switch}) with $\mathcal{U}_1=U$ given in Eq. (\ref{unitary_theta}) and $\mathcal{U}_2=CNOT$, we get the eigenvalues of the final state as shown in the Fig. (\ref{eval_theta_modW}). The eigenvalues are then arranged in non-increasing order for the whole range of the noise mixing parameter $p$ of the state and the unitary matrix parameter $\theta$. 
    
    \begin{figure}[htp]
    \fbox{\includegraphics[scale=0.5]{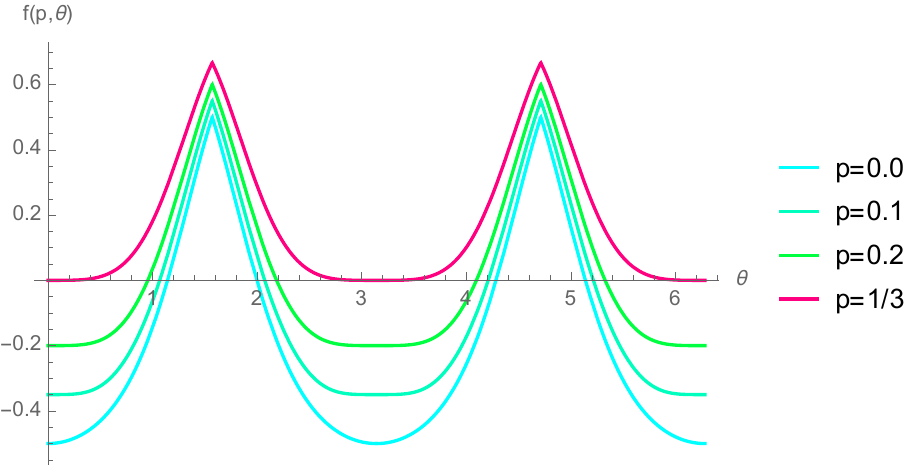}}
    \caption{\footnotesize{The LHS of the eigenvalue condition given in Eq. (\ref{AS_eigenvalue}) is plotted against the unitary matrix parameter for different values of state parameter of AS states.}}
    \label{LHS_eval_cond_modW}
    \end{figure}
    
    In the next step, to check whether the final state remains an AS or not we further evaluate the LHS of Eq. (\ref{AS_eigenvalue}) and obtain the expression as a function of $p$ and $\theta$ (say, $f(p, \theta)$). Plotting this function shows that the condition given in Eq. (\ref{AS_eigenvalue}) is not satisfied for a broad range of $p$ as well as $\theta$ as shown in Fig (\ref{LHS_eval_cond_modW}). Hence for a given choice of unitary it is possible to extend the range of $p$ for which the AS states become separable.

After proving that it is evidently possible to break absolute separability of modified Werner state via global unitary matrices while employing them with switching order, we look at the results a little bit more carefully. From Fig. (\ref{LHS_eval_cond_modW}) it is clear that for all the values of noise mixing parameter $p$ of the state, there exists certain type of global unitary (corresponding to given values of $\theta$ of $U$ given in Eq. (\ref{unitary_theta})) along with a CNOT gate, which can take out the AS state out of the convex set. Now, looking at Fig (\ref{eval_theta_modW}.(a)) we see that out of four eigenvalues of the final states, when one becomes $0$ for any value of $p$, another one also goes to $0$. At the same time, the highest eigenvalue takes the largest value. This might happen for several unitary matrices, one being $\theta=\pi/2$ in $U$. Correspondingly, the violation of Eq. (\ref{AS_eigenvalue}) becomes nothing but the largest eigenvalue. For $\theta=\pi/2$, the largest eigenvalue becomes $(1+p)/2$. Note that, the violation hence behaves as a monotonically increasing function of the noise mixing parameter $p$. Interesting even for $p=0$, when the corresponding initial state is a maximally mixed state, it is possible to have the condition for absolute separability violated for certain $\theta$. We emphasise here on the fact that the maximally mixed state $\mathcal{I}/4$ (in $2\otimes2$), which is a free state in every resource theory, it is possible to make it a separable state if one has access to global unitary operations making the initial state resourceful. On the other hand for $p=1/3$, the modified Werner states lie in the boundary of AS states, on the verge of becoming separable. Hence, intuitively it can be assumed that these are the states that are the easiest to take out from the set of convex states containing the AS states. Our result confirms the same as it can be seen from the plot in Fig. (\ref{LHS_eval_cond_modW}) that the choice of effective unitary is extensively high. 

\begin{figure}[htp]
\centering
\fbox{
\subfigure[For $p=0$ (maximally mixed state)]{\includegraphics[scale=0.155]{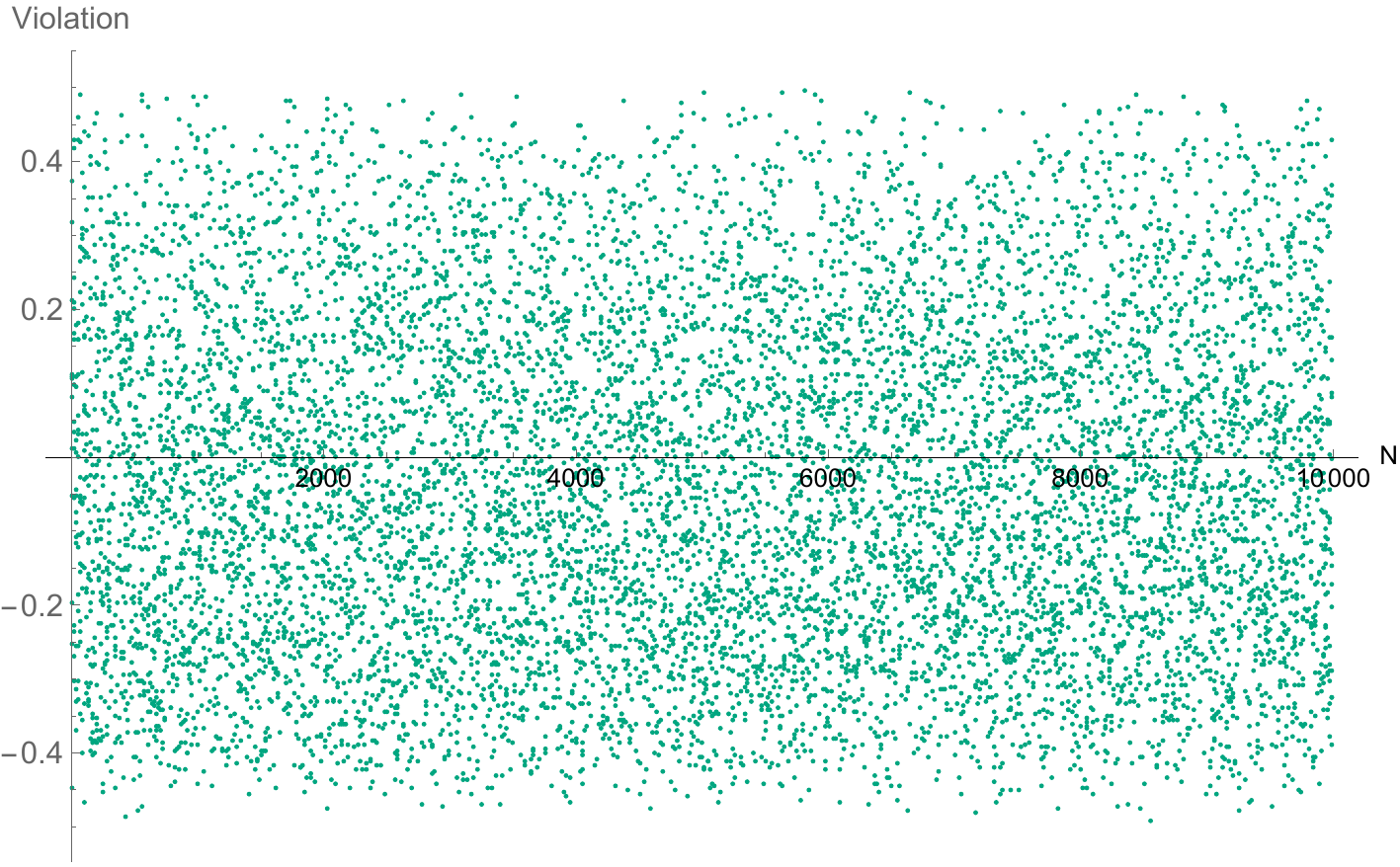}}
\qquad
\subfigure[For $p=1/3$, $\phi=0$, and $\gamma=\pi/4$]{\includegraphics[scale=0.155]{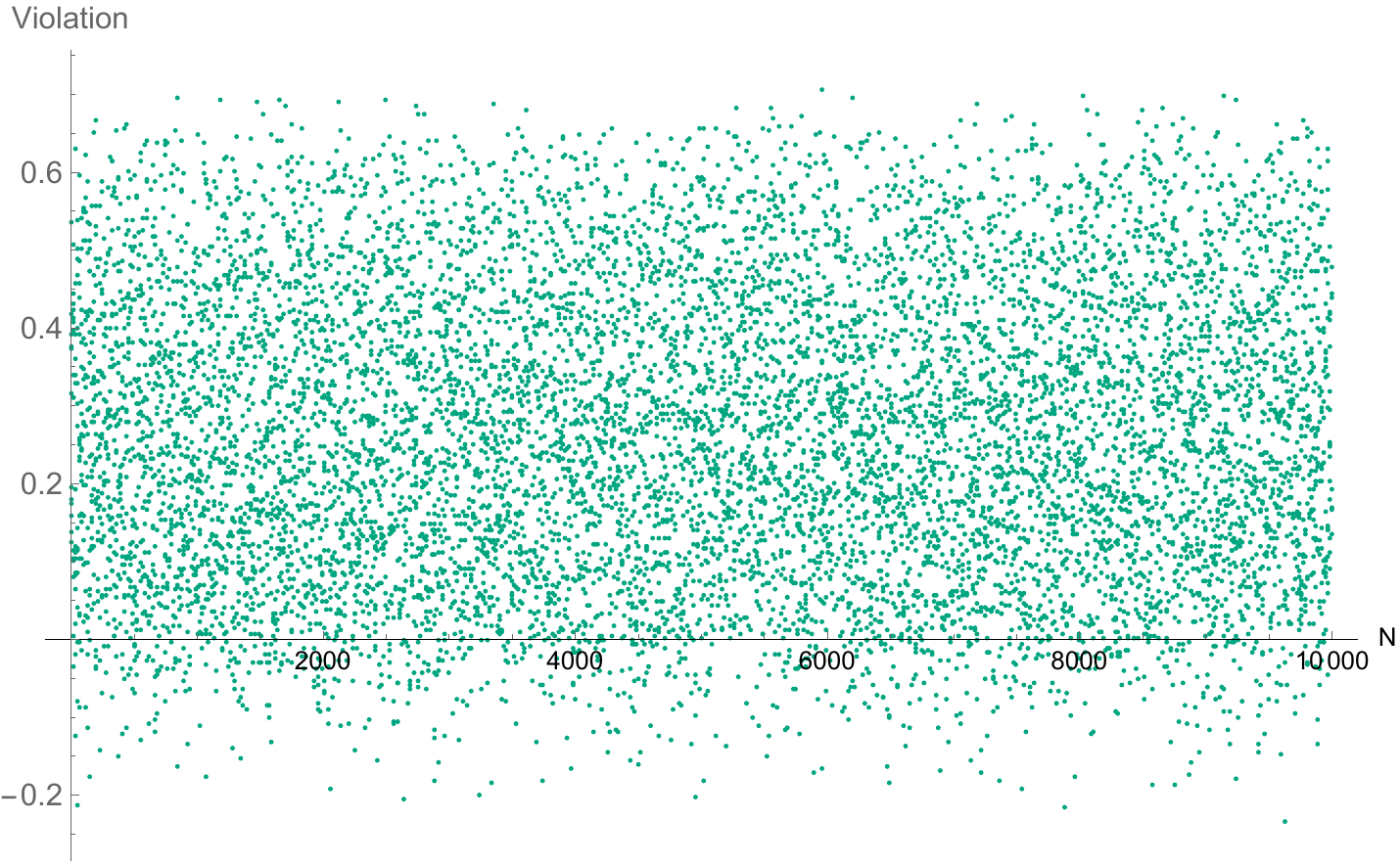}}}
\caption{\footnotesize{We plot the LHS of Eq.(\ref{AS_eigenvalue}) for $2\otimes2$ dimension (Violation) against the number of  Haar uniformly generated random unitary matrices (N) for two choices of noise mixing parameter $p$ for modified Werner state.}}
\label{RandomU_modW}
\end{figure}

Now to generalise our result numerically by considering that we have access to one CNOT gate and the other unitary is made completely random. We generate random $100000$ unitary matrices Haar uniformly. Here, we consider two cases corresponding to two initial states, (a) maximally mixed state in $2\otimes2$ dimension i.e. $p=0$ in Eq. (\ref{mod_Wstate}), and (b) $p=1/3$, $\phi=0$, and $\gamma=\pi/4$ in Eq. (\ref{mod_Wstate}), giving us a mixture of white noise with maximally entangled state in $2\otimes2$. For illustration, we plot the cases for $10000$ randomly generated unitary matrices in Fig.(\ref{RandomU_modW}). Note that, both the plots evidently show that it is always possible to find a global unitary operation that makes a AS state a separable one by switching operation with CNOT gate. As can be guessed intuitively, the number of useful unitary matrices to break absolute separability, for Case (b) is more than that for in Case (a) as the state in Case (b) is lying on the boundary of the convex set. We extend our result to higher dimension ($2\otimes d$) by considering maximally mixed state initially while both the random unitary matrices are generated Haar uniformly. We check the absolute separability criterion given in Eq. (\ref{AS_eigenvalue}) and plot the violation. The plots with more details can be found in the appendix.\\
Another interesting observation that can be made from the results of switching action on modified Werner state is about the rank of the final state. The initial state given in Eq. (\ref{mod_Wstate}) is a full-rank state for the whole range of noise mixing parameter $p$. After the switching action, the rank of the final state remains unaltered in most of the cases except for $\theta=\pi/2$ and for odd multiples of $\pi/2$. In these exceptional cases the rank decreases to 2 and making the final state certainly separable as can be seen from Proposition (\ref{prop1}). At the same time, from numerical simulation we find that rank remains unchanged for all the cases. 

\subsection*{Results on Bell diagonal states}

Next we move our discussion to the set of Bell-diagonal (BD) states. For that, we first introduce and elaborate the properties of this particular state in the following. The two-qubit Bell diagonal states are the probabilistic mixture of maximally entangled states given in the following form,
\begin{eqnarray}
    \rho_{BD} &=& p_1 \ket{\phi^+}\bra{\phi^+} + p_2 \ket{\phi^-}\bra{\phi^-} + p_3 \ket{\psi^+}\bra{\psi^+} \nonumber \\
    && + p_4 \ket{\psi^-}\bra{\psi^-}
    \label{BD_state_p}
\end{eqnarray}
with, $p_1+p_2+p_3+p_4=1$. Here, $\ket{\phi^{\pm}}=(\ket{00}\pm \ket{11})/\sqrt{2}$ and $\ket{\psi^{\pm}}=(\ket{01}\pm \ket{10})/\sqrt{2}$ are four Bell states. In the Bell-basis, this is a diagonal state. From the partial transposition criterion, it is evident that $\rho_{BD}$ is a separable state if all the mixing probabilities are less than or equal to $1/2$. Note that, the eigenvalues of the state in Eq. (\ref{BD_state_p}) are given by, $\{ p_1, p_2, p_3, p_4\}$. Another form of the state is \cite{LC_10},
\begin{equation}
    \rho_{BD} = \frac{1}{4} (\mathcal{I}+\sum_{i=1}^3 c_i \sigma_i \otimes \sigma_i)
    \label{BD_state_c}
\end{equation}
with $\sigma_i$'s being the Pauli matrices.

\begin{figure*}[htp]
\centering
\fbox{
\subfigure[Structure of Bell Diagonal states. The octahedron represents the set of separable states \cite{LC_10}, and the inner structure represents the AS states.]{\includegraphics[scale=0.3]{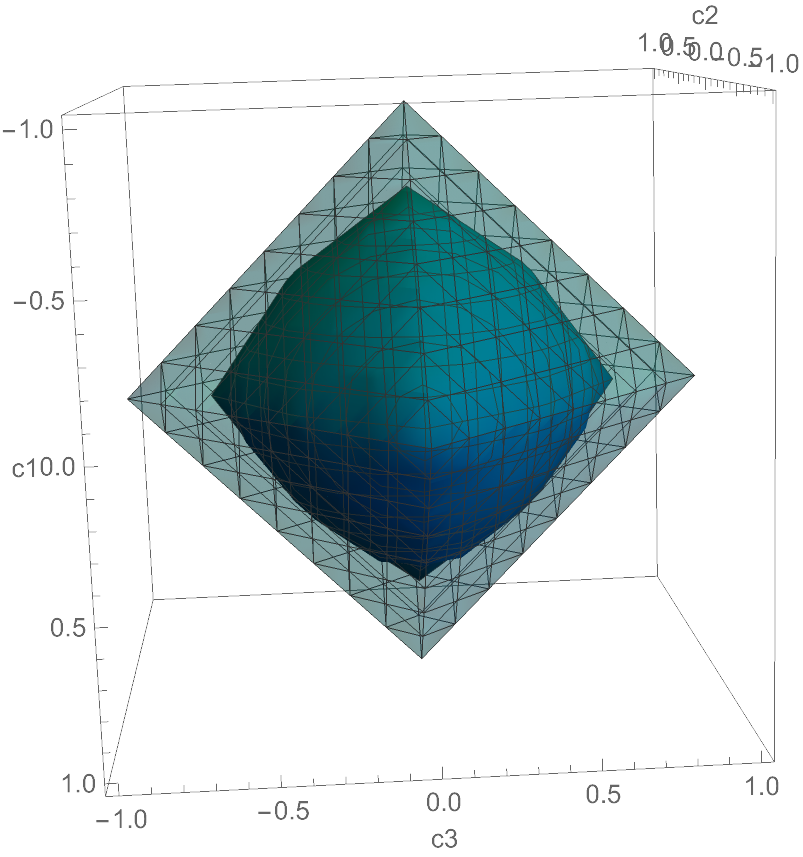}}
\qquad
\subfigure[Structure of BD AS states after switching action for Unitary corresponding to $\theta=\pi/6$]{\includegraphics[scale=0.2]{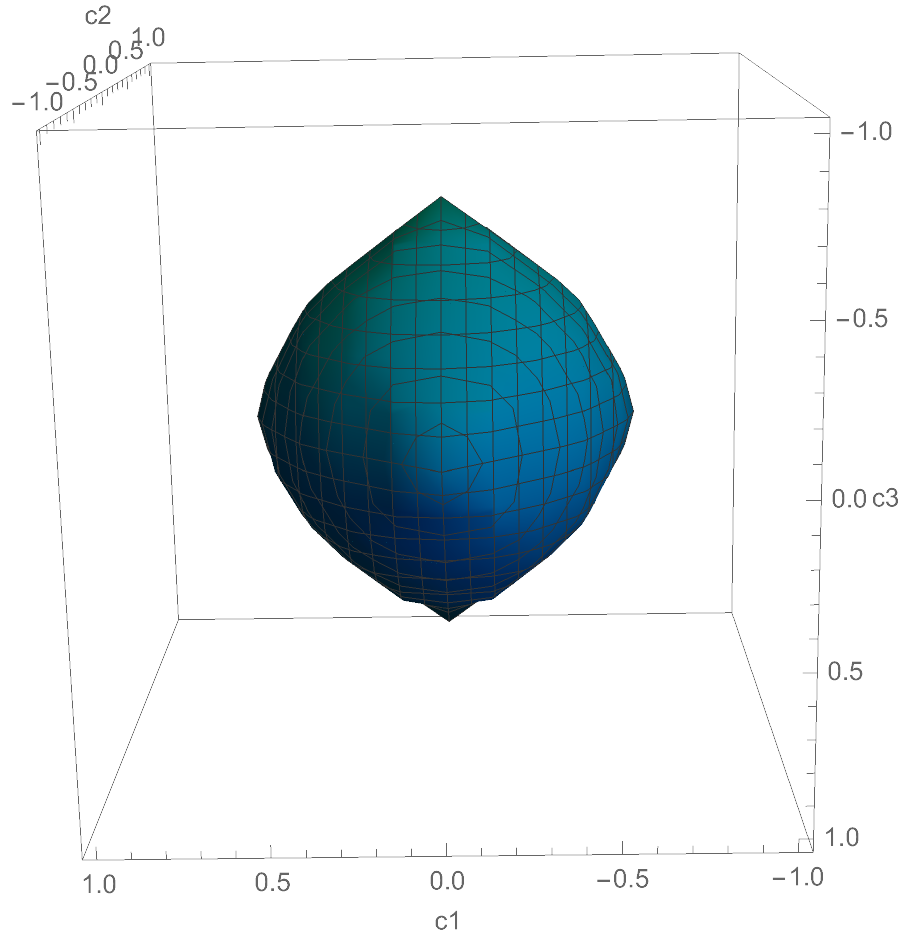}}
\qquad
\subfigure[Structure of BD AS states after switching action for Unitary corresponding to $\theta=\pi/4$]{\includegraphics[scale=0.2]{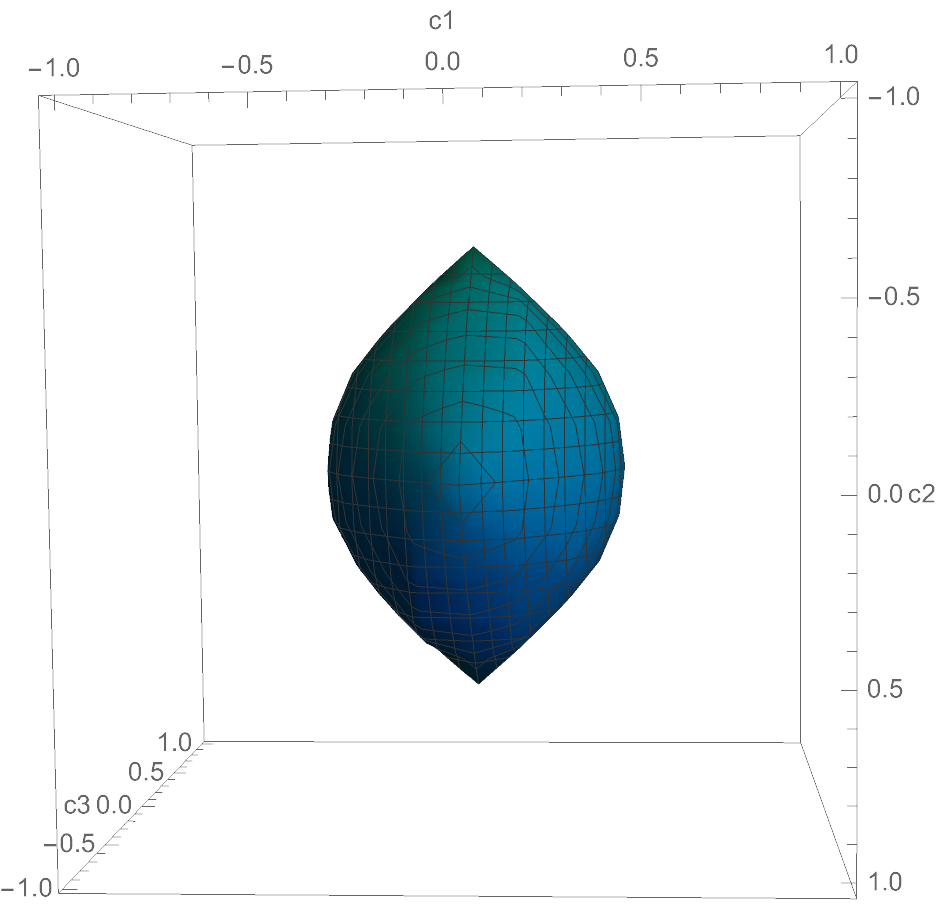}}
\qquad
\subfigure[Structure of BD AS states after switching action for Unitary corresponding to $\theta=\pi/3$]{\includegraphics[scale=0.2]{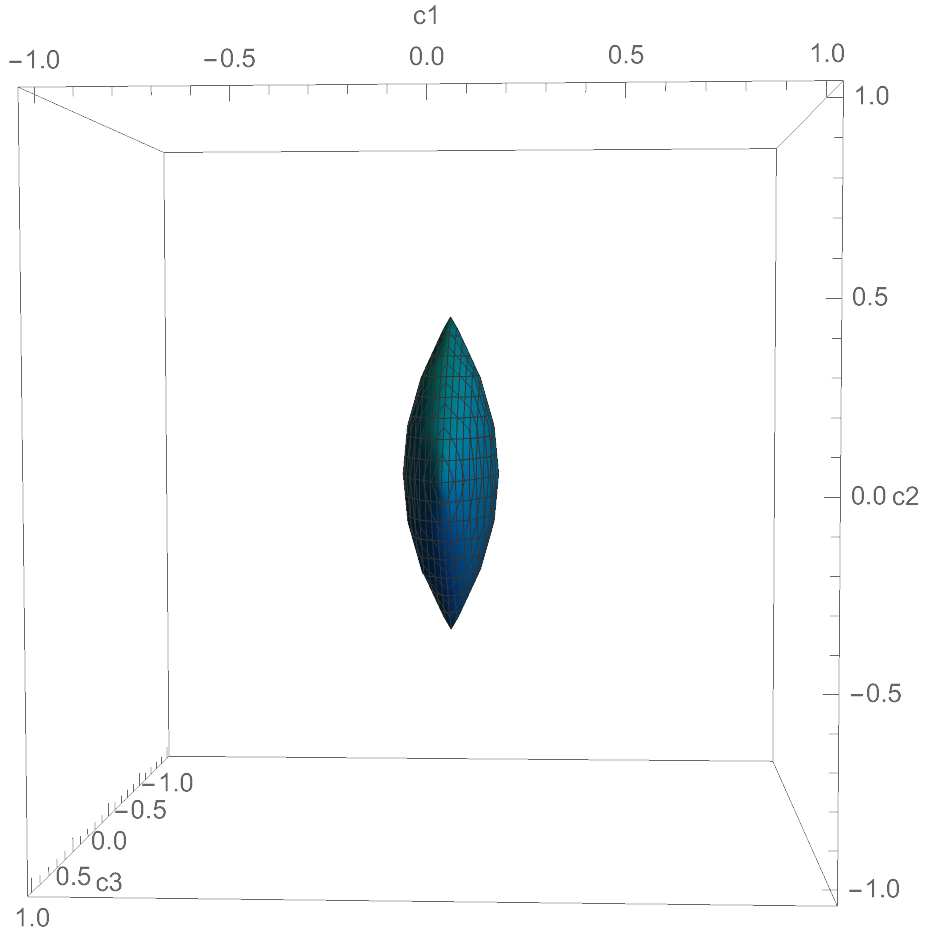}}
}
\caption{\footnotesize{Geometry of BD AS states before and after switching action of global unitary matrices.}}
\label{BD_AS_sep_switch}
\end{figure*}


\begin{figure*}[htp]
\centering
\fbox{
\subfigure[BD states for $\alpha=0.17$ to $0.5$. Here, $N=(\alpha-0.17)/0.01$.]{\includegraphics[scale=0.22]{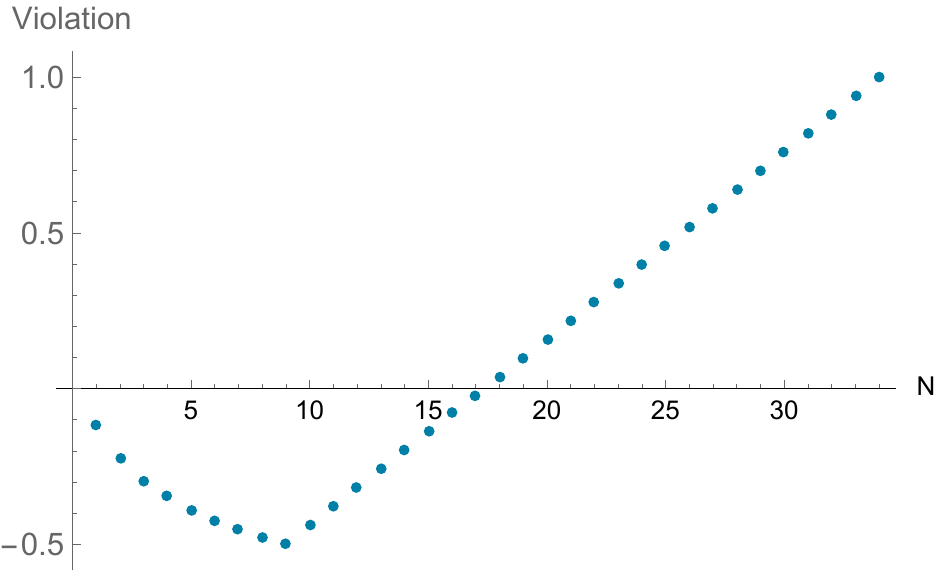}}
\qquad
\subfigure[For $\alpha=0.18$ after Switching action. (Here, $N$ is number of random unitary.)]{\includegraphics[scale=0.22]{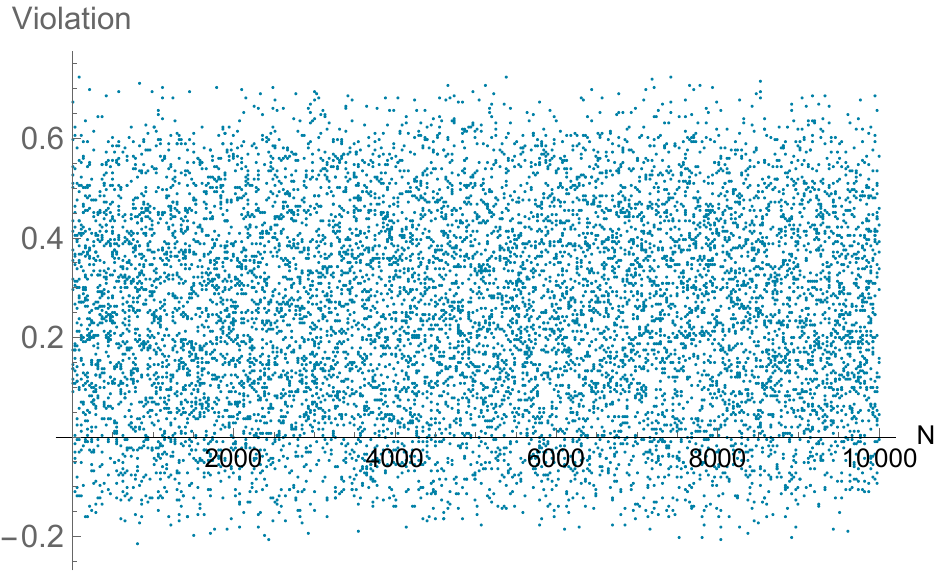}}
\qquad
\subfigure[For $\alpha=0.20$ after Switching action. (Here, $N$ is number of random unitary.)]{\includegraphics[scale=0.22]{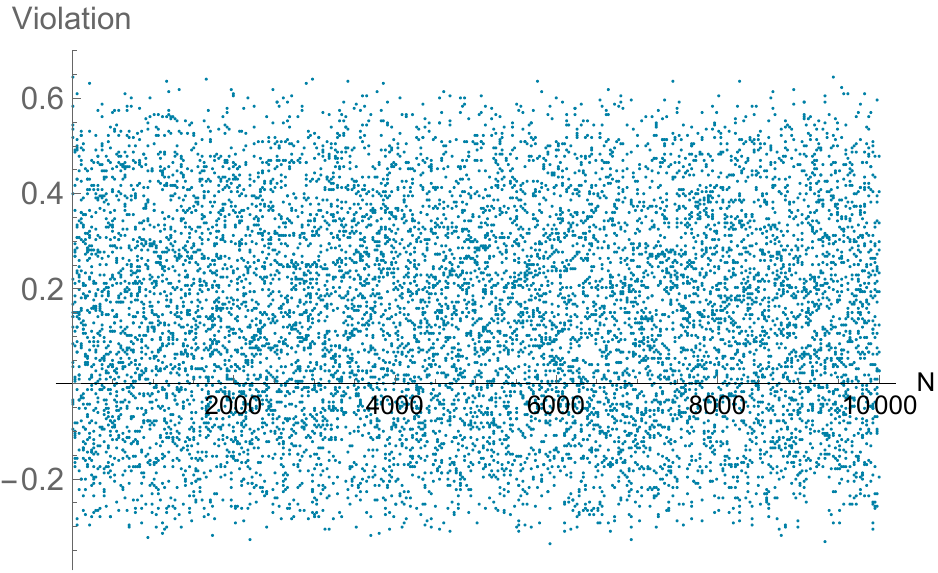}}
\qquad
\subfigure[For $\alpha=0.26$ after Switching action. (Here, $N$ is number of random unitary.)]{\includegraphics[scale=0.22]{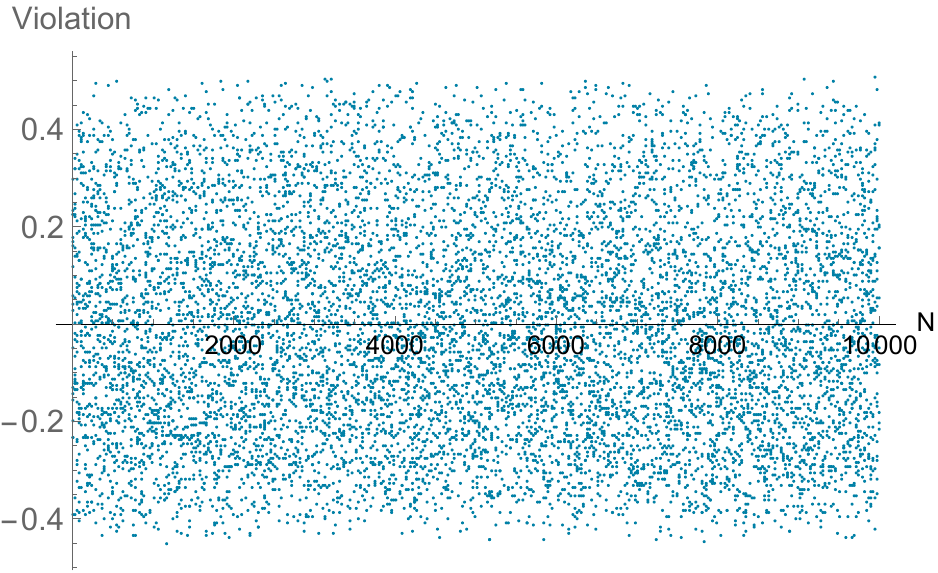}}}
\caption{\footnotesize{A particular type of BD state and Switching action of random unitary and CNOT gate on them.}}
\label{BD_random}
\end{figure*}

\noindent Comparing equations (\ref{BD_state_p}) and (\ref{BD_state_c}), one can get the constraints on the coefficients $c_1$, $c_2$ and $c_3$ which eventually indicates the region describing the separable states and the AS states \cite{LC_10}. For the state to be separable, we have $|c_1|+|c_2|+|c_3| \leq 1$ which gives a tetrahedron structure as illustrated in Fig. (\ref{BD_AS_sep_switch} (a)). Explicitly, $p_1=(1+c_1-c_2+c_3)/4$, $p_2=(1+c_1+c_2-c_3)/4$, $p_3=(1-c_1+c_2+c_3)/4$ and $p_4=(1-c_1-c_2-c_3)/4$. Putting the eigenvalue condition (\ref{AS_eigenvalue}) on these mixing parameters along with the condition of positivity, we find the conditions on $c_1$, $c_2$ and $c_3$ for $\rho_{BD}$ to be an AS state which enables us to obtain the structure of AS states as also illustrated in Fig. (\ref{BD_AS_sep_switch} (a)). 
\begin{result}
    If the initial state is $\rho_{BD}$ given in Eq.(\ref{BD_state_c}) and it undergoes the switching operation between two unitary channels, one being CNOT and the other given in Eq. (\ref{unitary_theta}), then it is possible to change the structure of the set of AS BD states for varying the parameter $\theta$ in Eq. (\ref{unitary_theta}) as shown in Fig. (\ref{BD_AS_sep_switch} (b), (c), (d)). 
\end{result}
For the understanding of the result we refer to Fig. (\ref{BD_AS_sep_switch}). Let us now note some observations that can be made from the mentioned figure. Firstly, the size of the set of AS BD states decreases with the increasing $\theta$ value. Eventually, it can been seen for $\theta=\pi/2$ the set completely vanishes which implies that all the BD AS states can be made atleast separable under suitable choice of global unitary matrices.\\
Now let us consider a particular case of BD state. We choose, three of the mixing parameter having the same value i.e. $p_2=p_3=p_4=\frac{1}{2}-\alpha$ making, $p_1=3\alpha-1/2$ with $\alpha$ running from $1/6$ to $1/2$. In this case, we plot the eigenvalue condition given in Eq. (\ref{AS_eigenvalue}) to find the range for AS states as shown in Fig. (\ref{BD_random}(a)). Then for some particularly chosen values of $\alpha$ we take CNOT gate along with a randomly generated unitary (Haar uniformly generated) and check the switching action on the given state as shown in Fig. (\ref{BD_random} (b), (c), (d)). It is clear from the plots that as $\alpha$ increases, the number of effective unitary matrices (though there always exists many) decreases. It is intuitive as from plot (a) in Fig. (\ref{BD_random}) the lower values of $\alpha$ are nearer to the boundary of the set of AS states. 
\subsection*{Results on higher dimensions}
\noindent  For the dimensions $2\otimes d$ we consider the maximally mixed state i.e. $\mathcal{I}_{2\otimes d}/2d$ as the initial state. This is definitely AS and moreover this is a free state in every resource theory. Then we Haar uniformly generate two unitary matrices and use them in switching action on the initial state. We explicitely deal with three cases of dimensions $2\otimes 3$, $2 \otimes 4$, and $2 \otimes 10$. The cases are plotted in Fig. (\ref{higher_d}). From the figure, it is evident that even in higher dimension where the condition given in Eq. (\ref{AS_eigenvalue}) holds, it is possible to find numerous number of global unitary matrices which can make AS states resourceful. Note that, the result in Proposition (\ref{prop2}) remains unchanged in these cases. 
\begin{figure*}[htp]
\centering
\fbox{
\subfigure[For $2\otimes 3$ dimension.]{\includegraphics[scale=0.32]{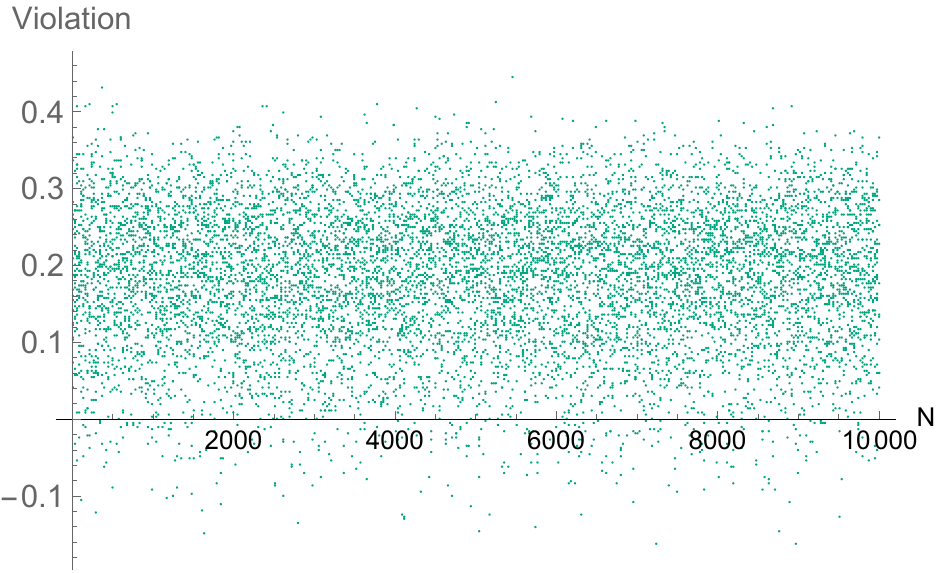}}
\qquad
\subfigure[For $2\otimes 4$ dimension.]{\includegraphics[scale=0.32]{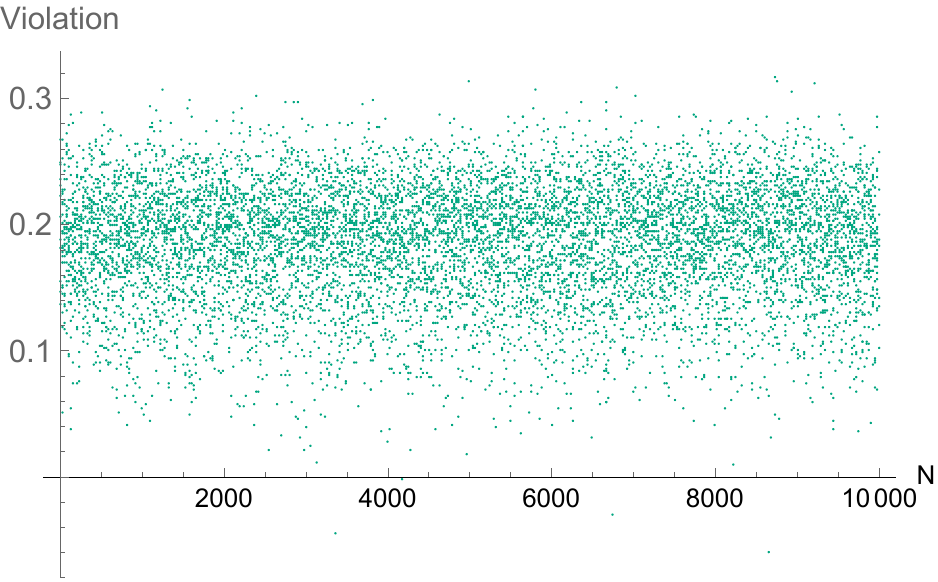}}
\qquad
\subfigure[For $2\otimes 10$ dimension.]{\includegraphics[scale=0.32]{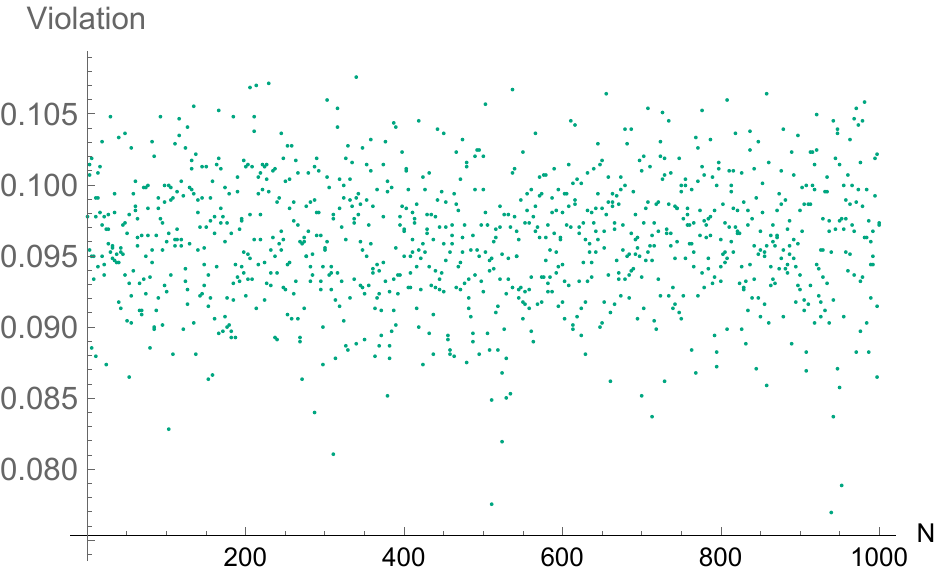}}}
\caption{\footnotesize{Breaking absolute separability in higher dimensions. The axes have their usual meaning.}}
\label{higher_d}
\end{figure*}



\section{Conclusion}

In the resource theory of non-absolute separability, global unitary matrices on the joint system represent free operations. Upon application of such operations, absolute separable states remain separable. The set of such states contains the maximally mixed state, which is treated as completely random and, hence, retains no viable information. In this paper, we show that it is possible to break absolute separability by switching action over global unitary operations. We start our study by considering a two-qubit quantum state residing on the boundary of the convex set of absolute separable states, and move on to more general versions of mixed absolute separable states. Furthermore, under suitable choices of the global unitaries, we show that it is possible to make even the maximally mixed state resourceful. Next, we move on to Bell diagonal states, where we outline the structure of AS states via the eigenvalue condition. Then, we characterize a particular class of BD states and establish that the switching action can take the AS states to the separable region. Our observation suggests that in most cases, the rank of the final state remains unchanged compared to the initial state, except for a few particular cases mentioned earlier. Along with the analytical studies over two-qubit AS quantum states, we have also done a thorough numerical study via randomly generated global unitaries acting on $2\otimes 3$, $2\otimes 4$ and $2\otimes 10$ dimensional states, validating our analytical results for higher dimensional quantum systems. Furthermore, following this direction, one can ask whether there is any constraint on the rank of the state for this switching action to be fruitful. Our study sheds light on the resourcefulness of quantum switch in the paradigm of entanglement resource theory, which in turn validates the seemingly vast applicability of switching action in quantum information science and technology in general.

\section*{Acknowledgement}

SG acknowledges partial support from the Department of Science and Technology, Government of India through the QuEST grant (grant number DST/ICPS/QUST/Theme-3/2019/120).

\bibliography{references.bib}

\end{document}